\documentclass[a4paper]{amsart}

\date{February 26, 2018}

\def\const{c}
\newcommand{\dbar}{{\mathchar'26\mkern-12mu\mathrm{d}}}
\def\cE{\mathcal{E}}
\def\cH{\mathcal{H}}
\def\gK{\mathfrak{K}}
\def\gR{\mathfrak{R}}
\def\gu{\mathfrak{u}}
\def\arctg{\mathrm{arctg}}
\newcommand{\rz}{\mathbb{R}}
\newcommand{\rd}{\mathrm{d}}

\def\gtf{\gamma_\mathrm{TF}}

\newtheorem{lemma}{Lemma} \newtheorem{theorem}{Theorem}
\title[Negative Ions]{The Maximal Negative Ion of the Time-Dependent
  Thomas-Fermi and the Vlasov Atom}
\author[L. Chen]{Li Chen} \address{Institut f\"ur Mathematik\\
  Universit\"at Mannheim\\ A5 6\\ 68131\\ Germnay}
\email{chen@math.uni-mannheim.de}

\author[H. Siedentop]{Heinz Siedentop} \address{Mathematisches
  Institut\\ Ludwig-Maximilians-Universi\"at M\"unchen\\
  Theresienstra\ss e 39\\ 80333 M\"unchen\\ Germany}
\email{h.s@lmu.de}

\begin{document}

\maketitle
\begin{abstract}
  We show an atom of atomic number $Z$ described by the time-dependent
  Thomas-Fermi equation or the Vlasov equation cannot bind more than
  $4Z$ electrons.
\end{abstract}
\section{Introduction and Statement of the Result}
\subsection{Known Results on the Excess Charge of Atoms}

Experimentally no doubly charged negative ions -- or ions that are
even more negative -- are known (Massey
\cite{Massey1976,Massey1979}). To prove this simple fact starting from
a mathematical model of the atom is called the excess charge
problem. Here the excess charge $Q(Z)$ refers to the maximal total
number of electron $N$ minus the nuclear charge $Z$. A step in this
direction was taken by Hill \cite{Hill1977,Hill1977P,Hill1980} in the
context of the Schr\"odinger equation showing that the $H^-$-ion has
only one bound state. First results on $Q$ itself were obtained by
Ruskai \cite{Ruskai1982,Ruskai1982II} and Sigal \cite{Sigal1982}
showing that atoms cannot be arbitrarily negative; later Lieb et al
\cite{Liebetal1984,Liebetal1988} showed that the excess charge is
asymptotically of lower order than $Z$, i.e., $Q(Z)/Z\to0$ as
$Z\to\infty$.  All of these results where obtained for Schr\"odinger
operators and are asymptotic for large $Z$.

For approximate models the results are -- not unexpected --
stronger. Solovej showed first for the reduced Hartree-Fock model
(Hartree-Fock without exchange term) and later for the full
Hartree-Fock model that excess charge $Q(Z)$ is uniformly bounded in
$Z$ (Solovej \cite{Solovej1991,Solovej2003}). These were big steps
forward, however, they are still short of the above mentioned fact,
that the observed excess charge is at most one, since no control on
the constant is offered.

For density functionals the situation is better. That there are
no negative ions in Thomas-Fermi theory is folklore (Gombas
\cite{Gombas1949}). This can be easily shown using a subharmonic
estimate (see Lieb and Simon \cite{LiebSimon1977}). Benguria and Lieb
\cite{BenguriaLieb1985} showed that the Thomas-Fermi-Weizs\"acker atom
can have an excess charge that does not exceed $0.7335$ (where this
numeric values holds for the coupling constant of the Weizs\"acker
term that reproduces the Scott conjecture).

Finally we wish to mention an unpublished discovery of Benguria in
Thomas-Fermi theory. He realized that multiplying the Thomas-Fermi
equation by $|x|\rho$ and integrating leads to the inequality
$$Q(Z)<Z$$
on the Thomas-Fermi excess charge. Of course this is of limited value
in TF theory, since, as mentioned above, the excess charge is
zero. However, the value of the idea is that it can be transferred and
extended to other situations. In fact, it was Lieb \cite{Lieb1984} who
realized this for the Schr\"odinger operator with and without magnetic
field and the Chandrasekhar operator. Lieb showed, among other things,
that for these operators $Q(Z)<Z+1$.  This bound holds regardless of
the symmetry under permutations (Boltzons, Bosons, or Fermions) and is
-- for large $Z$ worse than the asymptotic bounds mentioned
above. However, the bound is non-asymptotic and proves the ionization
conjecture for $Z=1$ (hydrogen).

Benguria's idea works also for the Hellmann functional, the
Hellmann-Weizs\"acker functional \cite{Benguriaetal1992}, the
Hartree-Fock functional and others. There are, however, functionals
that resisted such a treatment, like the
Thomas-Fermi-Dirac-Weizs\"acker functional and the M\"uller
functionals and variants thereof which were treated only recently by
different means (Frank et al \cite{Franketal2016,Franketal2018}, Kehle
\cite{Kehle2017}).  Moreover, there are models which are still
completely open like the no-pair operators of relativistic quantum
mechanics (Sucher \cite{Sucher1980}) where it is not even known that
the excess charge is finite.

The models mentioned so far were all treated in a stationary
setting. Although a time-dependent characterization of bound and
scattering states exists also in quantum mechanics (see Perry
\cite[Theorem 2.1]{Perry1983}), it took eighty-seven years after the
advent of quantum mechanics to approach the problem in a
time-dependent way (Lenzmann and Lewin
\cite{LenzmannLewin2013}). However, it turns out that this treatment
-- and therefore also ours since we will follow Lenzmann and Lewin as
closely as possible in the non-linear equation treated here -- is also
a variant of Benguria's original idea.

An improvement of Lieb's non-asymptotic result was obtained by Nam
\cite{Nam2012N} who showed
$$\alpha_N(N-1)\leq Z(1+0.68N^{-2/3})$$
with
$$\alpha_N:= \inf_{x_1,...,x_N\in\rz^3}
{\sum_{1\leq i<j\leq N}{|x_i|^2+|x_j|^2\over|x_i-x_j|}\over(N-1)\sum_{i=1}^N|x_i|}.$$

In this paper, we wish to discuss the excess charge problem of atoms
when described by the Vlasov equation \cite{Vlasov1938,Vlasov1968} and the
time-dependent Thomas-Fermi equation \cite{Bloch1933}. The latter is a
hydrodynamic one with a pressure-density relation given by
Thomas and Fermi.

\subsection{The Vlasov Equation}

The Vlasov equation, originally derived in plasma physics, can also be
used as an effective equation for the spin summed phase space density
$$
f:\rz\times \rz^{6}\to \rz_+
$$
of fermions.  If possible time dependence of the density is indicated
by a subscript $t$, i.e., the functions $f_t$ are interpreted as the
spin-summed phase space density of electrons at time $t$, at position
$x$ and momentum $\xi$. The Pauli principle for fermions with $q$ spin
states each (for electrons $q=2$) is implemented by the requirement
\begin{equation}\label{Linfty}
f_t\leq q.
\end{equation}
We will work in atomic units in which the rationalized Planck constant
$\hbar$ and the mass $m$ of the electron are one, in particular we have
$h=2\pi$. Following Planck one requires that each particle occupies
the volume $h^3$ in phase space, i.e., we interpret
$$\rho_t(x) := \int_{\rz^3} \dbar\xi f_t(x,\xi)$$
(with $\dbar \xi:=\rd \xi/h^3$) as the density of electrons at
position $x$ at time $t$ and
\begin{equation}\label{L1}
N(t):=\int_{\rz^3}\rd x \rho_t(x),
\end{equation}
as the number of particles at time $t$ which may or may not be finite.

We are interested in the maximal number of electrons which a nucleus
of charge $Z$ can bind. For the moment, though, we allow for arbitrary
many nuclei. The electric potential $V_\mathrm{tot}$ of $K$ nuclei at
positions $\gR_1,...,\gR_K$ with nuclear charges $Z_1,...,Z_K$ is
\begin{equation}
  \label{eq:potential}
  V_\mathrm{tot} := V - V_\mathrm{MF}:=\sum_{k=1}^KZ_k\delta_{\gR_k}*|\cdot|^{-1} -\rho*|\cdot|^{-1}.
\end{equation}
The force $\gK$ is
$$ \gK(x) := \nabla V_\mathrm{tot}(x)= -\sum_{k=1}^KZ_k{x-\gR_k\over |x-\gR_k|^3} 
+ \int_{\rz^3}\rd y\rho_t(y) {x-y\over|x-y|^3}.$$
Thus the Vlasov equation reads
\begin{equation}
  \label{VPN}
  \partial_t f_t + p\cdot \nabla_x f_t + \gK\cdot \nabla_\xi f_t  =0.
\end{equation}
 
For transparency we will assume for our main results that we are in
the atomic case, i.e., $K=1$, and $Z:=Z_1$, $\gR_1=0$. Using
homogeneity in the spirit of Benguria's idea, it is clear that
multiplying by a homogeneous function of degree one might be a hopeful
strategy; however, instead of multiplying simply by $x$, multiplying
by $x\cdot\xi|x|$ helps in dealing with the derivative with respect to
$\xi$. Because of the time-dependence, the obvious idea would be to
cut-off at an arbitrary distance $R$, integrate, and then take
$t\to\infty$. However because of technical reasons, a sharp cut-off
leads to an indefinite term later on. A suitable soft cut-off solves
this problem. And instead of taking $t$ large we will average over all
times. We will follow Lenzmann and Lewin \cite{LenzmannLewin2013} and
pick as test function
\begin{equation}
  \label{test}
  w_R:=\nabla g_R\cdot\xi,\ g_R(x) := R^3 g(|x|/R),\ g(r) = r-\arctg(r).
\end{equation}
We will show that the two potential terms will yield the wanted
estimate whereas the other terms of the equation vanish or can be
dropped.

\begin{theorem}
  \label{maxion}
  Assume $f_t$ to be a weak solution of the Vlasov equation
  \eqref{VPN} of finite energy \eqref{Energy}, assume $B\subset\rz^3$
  bounded and measurable, and set
  $$ N_V(t,B):= \int_{\rz^3} \dbar\xi \int_B\rd xf_t(x,\xi)$$
  which is the number of electrons in $B$.  Then in temporal average
  for large time $N_V(t,B)$ does not exceeds $4Z$, i.e.,
  \begin{equation}
    \label{behauptung}
    \limsup_{T\to\infty} {1\over T}\int_0^T\rd t N_V(t,B) \leq 4Z.
  \end{equation}
\end{theorem}

To interpret the result we introduce the following notation: we say
that $f_t \in L^1(\rz^3_x\times \rz^3_\xi)$ with $0\leq f_t\leq q$
a.e., is a fermionic bound state, if the functions $f_t$ fulfill the
following: for any $\epsilon>0$ exists a radius $R$ such that for all
times $t\geq0$
$$\int_{|x|>R} \rd x \dbar \xi f_t(x,\xi) < \epsilon.$$
Thus the theorem implies that the Vlasov equation has no bound state
which has more than $4Z$ electrons.

We remark, that it is obvious from the proof that the right side of
\eqref{behauptung} can be improved to $2Z$, if the spatial density
$\rho_t$ of the solution is radially symmetric.

Benguria's idea suggests to multiply the Vlasov equation with a
homogeneous function of degree one followed by a sharp cut-off and
integrate. Instead, to deal with the partial derivatives in $x$ and $\xi$, we
choose the weight $w$ given by
 \begin{equation}
   \label{w}
   w_R(x,\xi):= \nabla g_R(x)\cdot\xi= {|x|\over 1+(x/R)^2}x\cdot\xi,
 \end{equation} i.e.,
 $w(x)= x|x|\cdot\xi +O(|x|^3)$ for small $x$ and is bounded for large
 $|x|$ and fixed $R$ and $\xi$.
 
 To our knowledge no such result is known neither for the
 time-dependent Thomas-Fermi equation nor for the Vlasov
 equation. However, the analogue result for the Schr\"odinger equation
 was shown by Lenzmann and Lewin \cite{LenzmannLewin2013} whose proof
 we will follow as closely as possible. 

 \subsection{The Time-Dependent Thomas-Fermi Equation}
 The time dependent Thomas-Fermi equation (Bloch \cite{Bloch1933}), see
 also Gombas \cite{Gombas1949}), for electrons in the field of a
 nucleus $Z$ reads
 \begin{eqnarray}
   \label{eq:gtf}
   \partial_t\varphi_t=\frac12(\nabla \varphi_t)^2 +\int{\rd p\over\rho_t}
   - {Z\over|x|} + \rho_t*|\cdot|^{-1}  
 \end{eqnarray}
 supplemented by the continuity equation
 \begin{equation}
   \label{eq:kon}
   \partial_t\rho_t=\nabla(\rho_t\nabla\varphi_t).
 \end{equation}
 Here $\varphi$ is the potential of the velocity field $\gu$, i.e.,
 $\gu=-\nabla \varphi$, $\rho$ is the density of electrons, and $p$ is
 the pressure as a function of $\rho$. The Thomas-Fermi choice for $p$
 is $p(\rho):=(1/5)\gtf \rho^{5/3}$ where $\gtf:=(6\pi^2/q)^{2/3}$,
 i.e., we have
 \begin{equation}
   \label{eq:tf}
   \partial_t\varphi_t=\frac12(\nabla \varphi_t)^2 + \frac\gtf2\rho_t^{2/3}
   - {Z\over|x|} + \rho_t*|\cdot|^{-1}.
 \end{equation}

 Our result is
 \begin{theorem}
   \label{th:tf}
   Assume that $\varphi_t$ and $\rho_t$ is a weak solution of
   \eqref{eq:tf} and \eqref{eq:kon} with finite energy
   \eqref{eq:ttfe}, assume $B\subset\rz^3$ bounded and measurable, and
   set
  $$ N_\mathrm{TF}(t,B):= \int_B\rd x\rho_t(x)$$
  which is the number of electrons in $B$.  Then, in temporal average
  for large time, this does not exceed $4Z$, i.e.,
  \begin{equation}
    \label{behauptung-tf}
    \limsup_{T\to\infty} {1\over T}\int_0^T\rd t N_\mathrm{TF}(t,B) \leq 4Z.
  \end{equation}
 \end{theorem}

 We can interpret the result similarly to the Vlasov case: we say that
 a solution $(\varphi_t,\rho_t)$ fulfills the time dependent
 Thomas-Fermi equation \eqref{eq:tf} supplemented by \eqref{eq:kon} is
 a bound state of the Thomas-Fermi atom, if the solution
 $(\varphi_t,\rho_t)$ fulfills the following: for any $\epsilon>0$
 exists a radius $R$ such that for all times $t\geq0$
$$\int_{|x|>R} \rd x \rho_t(x) < \epsilon.$$
Thus, the theorem implies that the time-dependent Thomas-Fermi
equation has no bound state which has more than $4Z$ electrons. Again,
in the radially symmetric the constant reduces to $2Z$ as is obvious
from the proof.

\section{Uniform Estimates on Energies}

For the proof of our theorems we need some uniform estimate of the
kinetic energy. 

\subsection{Conservation of the total energy}
In this section we treat the general molecular case although not
needed in this generality for our result.
\subsubsection{The Vlasov Energy}
Suppose that $f_t$ is a weak solution of the Vlasov equation. Then it is
folklore that the energy
\begin{equation}\label{Energy}
  \cE_V(f_t) := \int_{\rz^3}\int_{\rz^3} \frac12 \xi^2f_t(x,\xi) \rd
  x\dbar\xi 
  -\int_{\rz^3}V(x)\rho_t(x) \rd x + D[\rho]+ R
\end{equation}
where
$$  D[\rho]:= \frac{1}{2}\int_{\rz^3}\!\!\int_{\rz^3}\frac{\rho_t(x)\rho_t(y)}{|x-y|} \rd x\rd y$$
is conserved, i.e., it is time independent.  Note that we added a
constant, namely the nuclear-nuclear repulsion
$$ R:= \sum_{0\leq k<l\leq K}{Z_kZ_l\over|\gR_k-\gR_l|}.$$

\subsubsection{The Thomas-Fermi Energy}
The time-dependent Thomas-Fermi energy is
\begin{equation}
  \label{eq:ttfe}
  \cH(\rho_t,\varphi_t):= \int_{\rz^3}\rd x{\rho_t(x)\over2}|\nabla\varphi_t(x)|^2 +\cE_\mathrm{TF}(\rho_t)
\end{equation}
where
$$\cE_\mathrm{TF}(\rho):=
:=\int_{\rz^3} \left(\frac3{10}\gamma_\mathrm{TF}\rho_t(x)^{5/3} - V(x)
\rho_t(x)\right) \rd x + D[\rho] + R.$$ The energy $\cH(\rho_t,\varphi_t)$ is
conserved along the trajectory of solutions $\varphi_t,\rho_t$.

\subsection{Lower Bound on the Energy}

We wish to show that the energy is bounded from below uniformly in
$f$. To this end we define the -- slightly non-standard -- spherical
symmetric rearrangement in the variable $\xi$
\begin{equation}
  \label{eq:umordnung}
  f^*(x,\xi) := q\chi_{B_{({6\pi^2\over q}\rho_t(x))^{1/3}}(0)}(\xi)
\end{equation}
where  $\rho(x):=\int_{\rz^3}\dbar\xi f(x,\xi)$; note
that also $\rho(x)=\int\dbar\xi f^*(x,\xi)$. Thus, obviously
$$T_V(f):=\frac12\int\rd x \int \dbar\xi \xi^2 f(x,\xi) \geq \frac12\int\rd x \int \dbar\xi \xi^2 f^*(x,\xi) 
= \frac3{10} \gamma_\mathrm{TF} \int\rd x \rho^{5/3}(x).$$
Thus, the Vlasov energy is bounded from below by the
Thomas-Fermi energy $\cE_\mathrm{TF}$
\begin{equation}
 \label{eq:TH}
  \cE_V(f_t) \geq \cE_\mathrm{TF}(\rho_t).
\end{equation}
This in turn is  bounded from below
by 
\begin{equation}
\label{TF}
\cE_\mathrm{TF}(\rho_t) \geq \alpha \sum_{k=1}^K Z_k^{7/3}.
\end{equation}
Note that this bound is uniform in the density $\rho$, therefore in
particular uniform in the electron number, and uniform in the
positions of the nuclei. (The proof uses Teller's lemma (Teller
\cite{Teller1962}, Lieb and Simon \cite{LiebSimon1977}), the scaling
of the minimum, and the fact that the excess charge of the
Thomas-Fermi functional vanishes.)  Here
$$\alpha := \inf\left\{ \int_{\rz^3} \left(\frac3{10}\gamma_\mathrm{TF}\rho(x)^{5/3} 
  - {\rho(x)\over|x|} \right)\rd x + D[\rho]\Big| \rho\geq0,\ \rho\in L^{5/3}(\rz^3),\
  D[\rho]<\infty\right\}.$$
In other words, we have shown stability of matter for the Vlasov
functional, i.e., whatever the initial conditions are, the energy is
bounded from below by a quantity that decreases at most linearly in
the number of involved atoms.

It is obvious that the analogous bound holds for the time-dependent
Thomas-Fermi theory, since by definition
$$ \cH(\rho,\varphi)\geq \cE_\mathrm{TF}(\rho).$$

\subsection{Upper Bounds on Norms Along the Trajectory of the Solution}

We show that the kinetic energy $T_{V}(f_t)$ and the Coulomb norm
$\|f_t\|_C:= \sqrt{D[\rho_{f_t}]}$ is uniformly bounded along the trajectory.

By the above $E(0):=\cE_V(f_0)=\cE_V(f_t)=:E(t)$, i.e.,
\begin{equation}
  \label{gleichmaessig}
  \begin{split}
    &\frac12(T_{V}(f_t) +\|f_t\|_C^2)\\
    \leq &E(0)+R - \frac12 \left(T_V(f_t) 
      +  \|f_t\|_C^2-\sum_{k=1}^K\int_{\rz^3}{2Z_k\rho_t(x)\over|x-R_k|}+4R\right)\\
    \leq &E(0) -\frac\alpha2\sum_{k=1}^K(2Z_k)^{7/3}
  \end{split}
\end{equation}
where we use the uniform lower bound \eqref{eq:TH} and \eqref{TF} on
the total energy.  In other words, both the kinetic energy $T_{f_t}$
and Coulomb norm $\|f_t\|_C$ are bounded along the trajectory
uniformly in time.

Again, the TF-case is similar.

\section{Proof of the Main Results}

\subsection{The Vlasov Case}
\begin{proof}[Proof of Theorem \ref{maxion}]
  First we note that -- of course -- $w \notin C_0^\infty(\rz^6)$.
  Strictly speaking we need to regularize the $w$ at the spatial
  origin and smoothly cutoff at infinity obtaining a weight
  $w_\epsilon$ which converges toward $w$ as $\epsilon\to0$. We have
  carried such procedure through in \cite{ChenSiedentop2017}. Since
  this is standard and only obscures the argument, we skip it.

  We do the previously announced: multiplication by $w_R$ defined in
  \eqref{test}, integration over phase space, and averaging in
  time. We get for the three summands $A$, $B$, and $C$ of the Vlasov
  equation:

  Summand \textbf{A}:
  \begin{align}
    |A| :=& \left|{1\over T}\int_0^T\rd t \int_{\rz^3}\rd x\int_{\rz^3}\dbar\xi\ w_R(x,\xi) \partial_tf_t(x,\xi)\right|\nonumber\\
    = &  \left|{1\over T}\int_{\rz^3}\rd x\int_{\rz^3}\dbar\xi\ w_R(x,\xi) f_T(x,\xi) 
        - \frac1T\int_{\rz^3}\rd x\int_{\rz^3} \dbar\xi\ w_R(x,\xi) f_0(x,\xi)\right| \label{A3}\\
    \leq& {1\over T}\left[\sqrt{T_{f_T}} \sqrt{\int_{\rz^3}\rd x |\nabla g_R(x)|^2f_T(x,\xi)}+ \sqrt{T_{f_0}}\sqrt{\int_{\rz^3}\rd x |\nabla g_R(x)|^2f_0(x,\xi)}\right]\label{A4}\\
    \leq& \const{ N^{1/2}R^2\over  T} \to 0\ \text{as}\ T\to\infty
  \end{align}
  where we used the Schwarz inequality to conclude line \eqref{A4}
  from line \eqref{A3}.

  \textbf{B}: First we mention that $g$ is a convex monotone increasing
  function which implies convexity of $g_R$, i.e., $\mathrm{Hess}(g_R)$
  is positive. Now, we integrate by parts
  \begin{align*}
    B := &{1\over T}\int_0^T\rd t \int_{\rz^3}\rd x \int_{\rz^3}\dbar
    \xi\
    \nabla g_R(x)\cdot\xi \xi\cdot\nabla_xf_t \\
    = &-{1\over T}\int_0^T\rd t \int_{\rz^3}\rd x 
    \int_{\rz^3}\dbar\xi\ \xi\cdot \mathrm{Hess}(g_R)(x)\xi\ f_t(x,\xi)\leq 0
  \end{align*}
  using the positivity of the Hessian in the last step.

  Eventually \textbf{C}:
  \begin{align}
    C:=& {1\over T}\int_0^T\rd t \int_{\rz^3}\rd x \int_{\rz^3}\dbar\xi
         \nabla g_R(x) \cdot\xi\ \gK\cdot\nabla_\xi f_t\\
    =& {1\over T}\int_0^T\rd t \int_{\rz^3}\rd x \int_{\rz^3}\dbar\xi\
       \nabla g_R(x)\cdot\xi \left(-Z{x\over|x|^3} + \int_{\rz^3}\rd y {x-y\over|x-y|^3}\rho_t(y)\right)
       \cdot\nabla_\xi f_t\\
    = & {1\over T}\int_0^T\rd t \left(\int_{\rz^3}\rd x Z{\nabla g_R(x)\cdot x\over|x|^3}\rho_t(x)\right. \\
       & \left.- \frac12\int_{\rz^3}\rd x \int_{\rz^3}\rd y {(\nabla g_R(x)-\nabla g_R(y))\cdot(x-y)\over|x-y|^3}\rho_t(x)\rho_t(y)\right)\\
    \label{22}
    \leq &  \frac1{T}\int_0^T\rd t\left(Z\underbrace{\int_{\rz^3}\rd x {\rho_t(x)\over \langle x/R\rangle^2}}_{=: M_R(\rho_t)}
           - \frac1{4}\int_{\rz^3}\rd x\int_{\rz^3}\rd y{\rho_t(x)\rho_t(y)\over \langle x/R\rangle^2\langle y/R\rangle^2)}\right)
  \end{align}
  where we used Lenzmann's and Lewin's Lemma \ref{ll} and the notation
  $\langle x\rangle:=\sqrt{1+|x|^2}$ in the last step. Thus, for fixed
  $R$

  $$0= A+B+C \leq \const {N^{1/2}R^2\over T} + \langle ZM_R(\rho_t)\rangle_T-\frac14\langle M_R(\rho_t))\rangle_T^2$$
  where we used Jensen's inequality to estimate the last term. Thus
  the temporal average up to $T\in[1,\infty]$,
  $\langle M_R(\rho_t)\rangle_T$, is uniformly bounded. Thus, as
  $T\to\infty$,
  \begin{equation}
    \label{punch}
    0= A+B+C\leq Z \langle M_R(\rho_t)\rangle_\infty -
    \frac14\langle M_R(\rho_t)\rangle_\infty^2
  \end{equation}
  where we set
  $\langle M_R(\rho_t)\rangle_\infty:=
  \limsup_{T\to\infty}T^{-1}\int_0^T\rd t \
  M_R(\rho_t)$. Furthermore, assume that $B$ is contained in the ball of
  radius $D\in\rz_+$ around the origin. Then we have
  \begin{align}
    &4Z\geq \langle M_R(\rho_t)\rangle_\infty\\
    =&\limsup_{T\to\infty}\frac1T\int_0^T\rd t \int_B\rd x{\rho_t(x)\over1+(x/R)^2}
      \geq \frac1T\limsup_{T\to\infty}\int_0^T\rd t \int_B\rd x{\rho_t(x)\over1+(D/R)^2}\\
    =& {1\over1+(D/R)^2}\frac1T\limsup_{T\to\infty}\int_0^T\rd t \int_B\rd x\rho_t(x)
  \end{align}
  Taking $R\to\infty$ on both sides gives the desired result
  $$4Z\geq
  \limsup_{T\to\infty}\frac1T\int_0^T\rd t \int_B\rd x \rho_t(x).$$
\end{proof}

\subsection{The Thomas-Fermi Case}
We now give the proof of the Thomas-Fermi case which initially
requires a new idea but towards the end is similar to the above proof.
\begin{proof}[Proof of Theorem \ref{th:tf}]
  We modify our strategy slightly: instead of multiplying
  \eqref{eq:tf} by the function $w_R$ we multiply it from the left by the operator
\begin{equation}
  \label{eq:W_R}
  W_R:=\nabla g_R \cdot \nabla,
\end{equation}
multiply by $\rho$, integrate in the space variable, and average in
time. The left side of \eqref{eq:tf} becomes
\begin{multline}
  \label{1}
  L_T:=\frac1T\int_0^T\rd t \int_{\rz^3}\rd x\rho_t\nabla g_R\cdot \nabla\partial_t\varphi_t\\
  = \frac1T\int_0^T\rd t \partial_t\int_{\rz^3}\rd x\rho_t\nabla g_R\cdot \nabla\varphi_t -\frac1T\int_0^T\rd t \int_{\rz^3}\rd x\partial_t\rho_t\nabla g_R\cdot \nabla\varphi_t.
\end{multline}
Since by the Schwarz inequality
\begin{multline}
  \label{gleichm}
  \left|\int_{\rz^3}\rho_T(x)\nabla g_R(x) \nabla\varphi_T(x)\rd x\right|\leq
  \|\nabla g_R\| \sqrt{\int_{\rz^3}\rho_T |\nabla\varphi_T|^2 \int_{\rz^3}\rho_T}\leq \const_R,
\end{multline}
we see that it is uniformly bounded in $T$ because of the analogue of
\eqref{gleichmaessig} for the time-dependent Thomas-Fermi equation
and the fact that the particle number is conserved in time. Thus, by
the continuity equation
\begin{multline}
  \label{2}
  \limsup_{T\to\infty}L_T = -\sum_{\mu,\nu=1}^3\limsup_{T\to\infty}\frac1T\int_0^T\rd t
  \int_{\rz^3}\rd x\partial_\mu(\rho_t\partial_\mu\varphi_t)
  \partial_\nu g_R \partial_\nu\varphi_t\\
  =\left\langle\sum_{\mu,\nu=1}^3
  \int_{\rz^3}\rd x\rho_t\partial_\mu\varphi_t
  \partial^2_{\mu,\nu} g_R \partial_\nu\varphi_t+\sum_{\mu,\nu=1}^3
  \int_{\rz^3}\rd x\rho_t\partial_\mu\varphi_t
  \partial_\nu g_R \partial^2_{\mu,\nu}\varphi_t\right\rangle_\infty\\
\geq \left\langle \sum_{\mu,\nu=1}^3
  \int_{\rz^3}\rd x\rho_t\partial_\mu\varphi_t
  \partial_\nu g_R \partial^2_{\mu,\nu}\varphi_t\right\rangle_\infty
\end{multline}
using integration by parts in the second but last step and the
positivity of $\mathrm{Hess}g_R$.

Next we treat the corresponding four resulting summands $R_1$ through
$R_4$ of the right hand side of \eqref{eq:tf}:
\begin{equation}
  \label{r1}
  R_1:= \left\langle \int_{\rz^3}\rd x\rho_t\nabla g_R\cdot \nabla\frac12(\nabla\varphi_t)^2\right\rangle_\infty=\left\langle\sum_{\mu,\nu=1}^3
  \int_{\rz^3}\rd x\rho_t\partial_\mu\varphi_t
  \partial_\nu g_R \partial^2_{\mu,\nu}\varphi_t\right\rangle_\infty
\end{equation}
which is identical to the last summand of the left side $\eqref{2}$.
\begin{multline}
  \label{r2}
  R_2:= \left\langle\int_{\rz^3}\rd x \rho_t\nabla g_R\cdot \nabla\frac{\gtf}{2} \rho_t^{2/3}\right\rangle_\infty\\
  = \frac15\gtf\left\langle\int_{\rz^3}\rd x \nabla \rho_t^{5/3}\cdot\nabla g_R\right\rangle_\infty = -\frac15\gtf \left\langle\int_{\rz^3}\rd x \rho_t^{5/3}\Delta g_R\right\rangle_\infty \leq 0
\end{multline}
again because of the positivity of $\mathrm{Hess}g_R$ and therefore of
$\Delta g_R$.
\begin{multline}
  \label{r3}
  R_3:= -\left\langle\int_{\rz^3}\rd x \rho_t\nabla g_R\cdot \nabla
    {Z\over|x|}\right\rangle_\infty = \left\langle\int_{\rz^3}\rd x
    \rho_t(x)
    {Z\nabla g_R(x)\cdot x\over|x|^3}\right\rangle_\infty\\
  = Z\left\langle \int_{\rz^3}\rd x {\rho_t(x)\over\langle
      x/R\rangle^2}\right\rangle_\infty =Z\left\langle
    M_R(\rho_t)\right\rangle_\infty
\end{multline}
using the notation of \eqref{22}.  Finally, the last summand in
\eqref{eq:tf} yields
\begin{equation}
  \label{r4}
  \begin{split}
    R_4:=& -\left\langle\int_{\rz^3}\rd x \int_{\rz^3} d y \nabla g_R(x) \rho_t(x)\rho_t(y){x-y\over|x-y|^3}\right\rangle_\infty\\
    =&-\frac12\left\langle\int_{\rz^3}\rd x \int_{\rz^3} d y
      \rho_t(x)\rho_t(y)
      {(\nabla g_R(x)-\nabla g_R(y))\cdot(x-y)\over|x-y|^3}\right\rangle_\infty\\
    \leq& -\frac1{4}\int_{\rz^3}\rd x\int_{\rz^3}\rd
    y{\rho_t(x)\rho_t(y)\over \langle x/R\rangle^2\langle
      y/R\rangle^2)}=-\frac14\left\langle M_R(\rho_t)\right\rangle^2_\infty
  \end{split}
\end{equation}
by \eqref{ll1}.

Putting the term obtained from the left side together with all the terms
obtained from the right side yields
\begin{equation}
  \label{eq:punch2}
  0\geq Z\left\langle M_R(\rho_t)\right\rangle_\infty
  - \frac14 \left\langle M_R(\rho_t)\right\rangle_\infty^2.
\end{equation}
The rest of the proof is now a mere copying of the Vlasov case.
\end{proof}

\appendix
\section{A Useful Inequality\label{a1}}

\begin{lemma}[Lenzmann and Lewin \cite{LenzmannLewin2013}]\label{ll}
For $g(r)=r-\arctg(r)$ we have
\begin{itemize}
\item For $x\neq y$ we have
  \begin{equation}
    \label{ll1}
    {(g'(|x|)\omega_x-g'(|y|)\omega_y)\cdot(x-y)\over|x-y|^3}\geq\frac12
    {g'(|x|)\over|x|^2}{g'(|y|)\over|y|^2}
  \end{equation}
\item Averaging over unit spheres yields
  \begin{equation}
    \label{ll2}
    \int_{\mathbb{S}^2}{\rd\omega_x\over 4\pi}\int_{\mathbb{S}^2}{\rd\omega_y\over4\pi}
      {(g'(|x|)\omega_x-g'(|y|)\omega_y)\cdot(x-y)\over|x-y|^3}\geq
      {g'(|x|)\over|x|^2}{g'(|y|)\over|y|^2}.
    \end{equation}
\end{itemize}
\end{lemma}

Note also the related inequality (Lenzmann and Lewin \cite[Lemma
3]{LenzmannLewin2013} for $\nu=3$) and Chen and Siedentop
\cite{ChenSiedentop2017} for general $\nu$) for $x\neq y\in\rz^\nu$
\begin{equation}
 {(|x|^{\nu-1}\omega_x-|y|^{\nu-1}\omega_y)\cdot(x-y)\over|x-y|^\nu}\geq 2^{2-\nu}.
\end{equation}

{\small \textit{Acknowledgment:} We acknowledge support by the
  Deutsche Forschungsgemeinschaft through the grants CH 955/4-1 and SI
  348/15-1.}

\def\cprime{$'$}

\end{document}